\documentclass[
10pt,
pra,
twocolumn]
{revtex4-1}

\usepackage{graphicx}
\usepackage{amsmath, amsthm, amssymb}
\usepackage{subfigure}
\usepackage{comment}
\usepackage{bm}
\usepackage{epsfig,color}
\usepackage[sans]{dsfont}

\newcommand{\tr}{\text{tr}}
\newcommand{\ket}[1]{| #1 \rangle}
\newcommand{\bra}[1]{\langle #1|}

\newcommand{\be}{\begin{equation}}
\newcommand{\ee}{\end{equation}}
\newcommand{\bea}{\begin{eqnarray}}
\newcommand{\eea}{\end{eqnarray}}
\newcommand{\bes}{\begin{equation*}}
\newcommand{\ees}{\end{equation*}}
\newcommand{\beas}{\begin{eqnarray*}}
\newcommand{\eeas}{\end{eqnarray*}}

\newcommand{\rank}{\operatorname{rank}}
\newcommand{\range}{\operatorname{range}}

\newcommand{\proj}[1]{\ket{#1}\!\bra{#1}}

\renewcommand{\H}{\mathcal{H}}

\def\B{\mathcal{B}}

\def\tr{\mathrm{tr}}
\def\T{\mathcal{T}}

\def\i{\mathrm{id}}
\def\sn{\mathrm{SN}}

\def\m{\Lambda_{\mathrm{TP}}^{(k)}}
\def\ma{\Lambda_{\mathrm{TA}}^{}}

\newtheorem{thm}{Theorem}[section]
\newtheorem*{thm*}{Theorem}

\newtheorem{lem}[thm]{Lemma}
\newtheorem*{lem*}{Lemma}
\newtheorem{prop}[thm]{Proposition}

\newtheorem*{lipschitzLem*}{Lemma \ref{lipschitz}}
\newtheorem*{lipschitzCubeLem*}{Lemma \ref{lipschitzCube}}
\newtheorem*{pgmNearlyOptimalThm*}{Theorem \ref{pgmNearlyOptimal}}

\begin{document}

\title{ More entanglement implies higher performance in tailored channel discrimination tasks  }

\date{\today}

\author{Joonwoo Bae}
\affiliation{ School of Electrical Engineering, Korea Advanced Institute of Science and Technology (KAIST), 291
	Daehak-ro, Yuseong-gu, Daejeon 34141, Republic of Korea}

\author{ Dariusz Chru\'sci\'nski}
\affiliation{Institute of Physics, Faculty of Physics, Astronomy, and Informatics, Nicolaus Copernicus University, Grudziadzka 5, 87-100 Torun, Poland  }

\author{Marco Piani}
\affiliation{SUPA and Department of Physics, University of Strathclyde, Glasgow G4 0NG, UK}

\begin{abstract}
We show that every entangled state provides an advantage in bi- and multi-channel discrimination that singles out its degree of entanglement, quantified in terms of its Schmidt number and of the corresponding robustness measures.
\end{abstract}

\maketitle

Entanglement~\cite{horodecki2009quantum} is a resource in quantum information processing \cite{nielsen2010quantum}. Entangled states cannot be shared by arbitrarily many parties~\cite{coffman2000} and thus are naturally a resource for secure communication \cite{ref:crypto}. The preparation of large-size entangled states allows one to run quantum algorithms by means of local measurements \cite{ref:mbqc}. Entanglement also allows one to outperform classical counterparts in various information applications. For example, it is useful in improving measurement precision, i.e. in quantum metrology \cite{ref:qmetro}. In fact, all entangled states, even when only slightly entangled, are useful for some task, such as in teleportation \cite{ref:masanes}. Moreover, there are entangled states which are not distillable yet can still lead to an increase in channel capacity \cite{ref:smith}, can be used in quantum crytography \cite{boundentkey}, and also lead to producing probabilities that cannot be prepared by shared randomness and local operations \cite{ref:brunner}. 

A process in quantum information theory, that is, some dynamics of quantum states, corresponds to a quantum channel described in mathematical terms by a so-called trace-preserving and completely positive map over quantum states. A linear map on operators is said to be positive if it preserves the positive semidefiniteness of operators, while it is called completely positive if it remains positive under an arbitrary extension to a an environment or reference system. This distinction is closely related to the existence of entangled states, and not needed in the classical case: the classical analogue of a channel, that is, the most general stochastic process, is simply identified by a positive transformation on probability distributions, since positivity itself already implies the condition of being completely positive. For any positive but not completely map on quantum states, however, there exist entangled states of system and environment such that under the partial action of the map, the entangled states are transformed into non-positive operators, that cannot be interpreted as quantum states \cite{ref:horodeckif}. 

A related property is that, in contrast to the classical analogy of stochastic processes, correlations between a probe and an ancilla can be useful in improving the distinguishability of quantum channels \cite{ref:dariano}. Recently, a tight connection between entanglement and quantum channel discrimination has been established: for any entangled state, independently of how weakly entangled, there exist a pair of quantum channels for which the state is useful to improve channel distinguishability with respect to the best discrimination strategy that does not make use of correlations \cite{ref:pw}. In other words, all entangled states are useful in quantum channel discrimination, or more precisely, in some tailored (to the state under scrutiny) quantum channel discrimination task. The result is remarkable since distinguishing quantum evolutions is a fundamental and operational task in information theory (see e.g.\cite{ref:renner} for applications), and shows the equivalence between the conditions of entanglement and some improvement in channel distinguishability.

The advantage that entanglement can provide in the discrimination of quantum evolutions has been further scrutinized, e.g. with respect to the issue of whether the advantage is preserved if constraints are imposed on the measurements that can be performed after the evolution has taken place~\cite{matthews2010entanglement,piani2015necessary}. In this work, we prove several refinements of the results of Ref.~\cite{ref:pw}, not with respect to the issue of quality of the measurement, but of the quality of the entangled input. In particular, we show that every entangled state provides an advantage in bi- and multi-channel discrimination that singles out its degree of entanglement, quantified in terms of its Schmidt number~\cite{terhal2000schmidt} and of the related robustness measures~\cite{vidal1999robustness}. 

\section{Preliminaries}

Throughout, we denote by  $\H_d$ a Hilbert space of finite dimension $d$. The set of bound linear operators in $\H_d$ is denoted by $\B(\H_d)$, and the subset of quantum states, i.e. the set of positive semidefinite operators having unit trace, by $S(\H_d)$. 

\subsection{Entanglement}
When considering tensor-product systems, e.g. with Hilbert space $\H_d\otimes \H_{d'}$, one typically distinguishes between product states $\ket{\alpha}_{\H_d}\otimes\ket{\beta}_{\H_{d'}}$ and entangled states that are not product. One can have a finer classification based on the Schmidt decomposition~\cite{nielsen2010quantum}; for an arbitrary state $\ket{\psi}_{\H_d \H_{d'}}$ it is possible to identify local orthonormal bases $\{\ket{a_i}_{\H_d}\}$ and $\{\ket{b_i}_{\H_{d'}}\}$ such that
\begin{equation}
\label{eq:SD}
\ket{\psi}_{\H_d \H_{d'}} = \sum_{i=1}^{\mathrm{SR}(\ket{\psi})} \sqrt{q_i}\ket{a_i}_{\H_d}\otimes \ket{b_i}_{\H_{d'}},
\end{equation}
where the Schmidt rank $\mathrm{SR}(\ket{\psi})$ indicates the number of non-vanishing tensor-product terms. For a mixed state $\rho$ in $S(\H_d\otimes \H_{d'})$, the Schmidt number is then defined as~\cite{terhal2000schmidt}
\bea
\mathrm{SN}(\rho) = \min_{\{ p_i , | \psi_i \rangle \}} \max_{i } ~\mathrm{SR}(|\psi_i \rangle), \label{eq:sn}
\eea
with the minimum taken with respect to arbitrary pure-state decompositions satisfying $\rho = \sum_i p_i \ket{\psi_i}\bra{\psi_i}$. Let $S_k$ denote the set of states having Schmidt number not greater than $k$. Such a set is clearly convex and compact; more explicitly, it is clear that such a set is given by the convex hull of all pure states that have Schmidt number not greater than $k$. 
We have a strict order among the sets, such that
\bea
S_{1} \varsubsetneq S_2 \varsubsetneq \cdots \varsubsetneq S_{d_\textrm{min}} = S(\H_A \otimes \H_S), \label{eq:sets}
\eea
where $d_{\textrm{min}}$ denotes the dimension of the Hilbert space of the smaller of the two entangled systems. Note that $S_1$ corresponds to the set of separable states of the form $\sum_i p_i \ket{\alpha_i}\bra{\alpha_i}\otimes\ket{\beta_i}\bra{\beta_i}$, for which correlations can be explained in classical terms.
 
\subsection{Witnesses and robustness}

Whenever a bipartite state $\rho\in S(\H_d\otimes \H_{d'})$ has Schmidt number strictly larger than $k$ (with $k<\min\{d,d'\}$), there is a witness $W^{(k)}\in \B(\H_d\otimes \H_{d'})$ such that~\cite{sanpera2001schmidt,skowronek2009cones,horodecki2009quantum,bae2016operational}
\[
\tr(\rho W^{(k)}) < 0
\]
while
\[
\tr(\sigma W^{(k)}) \geq 0\quad \forall \sigma\in S_k.
\]
In order to quantify to what extent a bipartite  state $\rho$ fails to be part of the subset $S_k$, one can define a generalized robustness measure $R_{S_k}$, originally introduced just for entanglement per se~\cite{vidal1999robustness}:
\begin{equation}
\begin{split}
&R_{S_k}(\rho)\\
&:= \min\{R\geq 0 \, |\\
&\qquad\rho = (1+R)\sigma^{(k)}-R\tau, \sigma\in S_k, \tau\in S(\H_d\otimes \H_{d'})\}.
\end{split}
\end{equation}
This quantifier can be given the interpretation of minimal noise needed to destroy the ``resource'' constituted by $\rho$ for not being in $S_k$, since one is asking for the minimal noise rate $R\geq 0 $ such that
\[
\frac{\rho + R \tau}{1+R} \in S_k
\]
for some general bipartite state $\tau$. It was soon recognized that the robustness admits an expression in terms of ``quantitative witnesses''~\cite{brandao2005quantifying}. In general, for the Schmidt-number robustness $R_{S_k}$, one has the following semidefinite programming (SDP)~\cite{boyd2004convex,watrous2018theory} expression 
\begin{equation}
\begin{aligned}
R_{S_k}(\rho) =	{\text{maximize}}\quad & - \tr(W^{(k)}\rho) \\
	\text{subject to}\quad \label{eq:SDPcondition1} &   \tr(W^{(k)} \sigma^{(k)})\geq 0\quad \forall \sigma^{(k)}\in S_k\\
	& W^{(k)}\leq \openone.
\end{aligned}
\end{equation} 
Obviously, there is a hierarchy
\[
R_{S_k}(\rho)\leq R_{S_{k'}}(\rho)\textrm{ for } k\geq k',
\]
but notice that, given two bipartite states $\rho$ and $\sigma$, it might happen that $R_{S_k}(\rho) < R_{S_k}(\sigma)$ while $R_{S_{k'}}(\rho) > R_{S_{k'}}(\sigma)$, for $k\neq k'$. It is easy to construct such examples even just considering pure states.

\subsection{Quantum operations}

We denote by $\mathcal{T}(\H_d , \H_{d'})$ the set of linear mappings from $\B(\H_d)$ to $\B(\H_{d'})$. A mapping $\Lambda\in \T(\H_{d} , \H_{d'})$ is called Hermitian (also known as Hermiticity-preserving) if for all $X\in\B(\H_d)$ it holds that $\Lambda^{\dagger}(X) = \Lambda(X^{\dagger})$. A linear mapping $\Lambda$ is trace-nonincreasing if $\tr \Lambda (X) \leq \tr X$ for all positive semidefinite $X\geq0$, trace-preserving if $\tr \Lambda (X) = \tr X$ for all $X\in\B(\H_d)$, and trace annihilating if instead $\tr \Lambda (X) = 0$ for all $X$. A mapping is called positive (also known as positivity-preserving) if $\Lambda (X) \geq 0$ for all positive semidefinite operators $X\geq 0$. Then, a map $\Lambda$ is called $k$-positive if it remains positive under extension to $k$-dimensional environment, i.e. $\i_k \otimes \Lambda$ is positive, with the identity map $\i_k \in \T (\H_{k} , \H_{k})$. If a map is $k$-positive for all $k\geq 1$, then it is called completely positive. A quantum channel is described by a trace-preserving and completely positive linear map on quantum states. Note that a map in $\T(\H_d , \H_{d'})$  is completely positive if and only if it is $d$-positive.

An instrument is a collection of completely positive trace-nonincreasing linear maps $\{\tilde{\Lambda}_i\}$, so-called subchannels, such that $\Lambda = \sum_i \tilde{\Lambda}_i$ is a channel. The notion of instrument captures mathematically the concept of branching of a linear evolution of a system, e.g., it may describe how an atom undergoes a de-excitation process or does not, but it may also describe some stochastic controlled evolution. 
An instrument allows one to calculate both the (state-dependent) probability of the different branches and the corresponding final state of the system: in general, subject to the action of an instrument $\{\tilde{\Lambda}_i\}$, a quantum system initially in a state $\rho$ evolves into a (renormalized) state $\rho' = \tilde{\Lambda}_i (\rho)/\tr \tilde{\Lambda}_i(\rho)$ with probability $\tr \tilde{\Lambda}_i(\rho)$. A special case of instrument is that of multiple channels $\{\Lambda_i\}$, applied according to a priori probability distribution $\{p_i\}$, in which case one can think of the instrument $\{\tilde{\Lambda}_i = p_i \Lambda_i\}\}$. Notice that in the latter case, each subchannel is not only trace-nonincreasing, but trace-scaling.

\subsection{Choi-Jamiolkowski isomorphism}

There exists a natural mapping, actually an isomorphism, between linear maps and quantum operators, that is, between $\T(\H,\H')$ and $\B(\H\otimes \H')$~\cite{ref:jiso,ref:choi}. For a given map $\Lambda \in \T (\H, \H')$, for $\H=\H_d$, the corresponding unique element in $\B(\H\otimes \H')$ is given by
\[
\chi_{\Lambda} = (\i \otimes \Lambda) (|\phi^{+}\rangle \langle \phi^{+} |),
\]
with $|\phi^{+}\rangle = (\ket{1}\ket{1} + \cdots + \ket{d}\ket{d}) / \sqrt{d}$ a fixed standard maximally entangled state, where $\{|i\rangle\}_{i=1}^d$ denotes some fixed orthonormal basis of $\H$. The operator $\chi_{\Lambda}$ is called the Choi matrix of the map $\Lambda$; for a quantum channel $\Lambda$, its Choi matrix is a state, called the Choi-Jamiolkowski (CJ) state. A map $\Lambda \in \T (\H_k,\H_{k'})$ is trace-preserving if and only if $\tr_{\H_{k'}} ( \chi_{\Lambda} ) = \mathrm{I}_{\H_{k}}/k$ where $\mathrm{I}_{\H}$ denotes the identity operator on $\H$, and trace-nonincreasing if and only if $\tr_{\H_{k'}} ( \chi_{\Lambda} ) \leq \mathrm{I}_{\H_{k}}/k$. On the other hand, a map $\Lambda \in \T (\H_k,\H_{k'})$ is unital if and only if $\tr_{\H_{k}} ( \chi_{\Lambda} ) = \mathrm{I}_{\H_{k'}}/k$.

\subsection{Schmidt number and $k$-positive maps}

It holds that, for a state $\rho_{AS}$, its Schmidt number is not larger than $k$ if and only if it remains positive under the partial action of all $k$-positive maps~\cite{terhal2000schmidt}, i.e.
\begin{multline}
\label{eq:iffSk}
\rho_{AS}\in S_k
\\
\iff \i \otimes \Lambda (\rho_{AS}) \geq 0 \textrm{ for all }k\textrm{-positive maps } \Lambda^{(k)}. 
\end{multline}
Clearly, the set of separable states, for which the Schmidt number is $k=1$, corresponds to the states which remain positive for all positive maps. On the other hand, by definition a map is $k$-positive if and only if, through its partial action, it preserves the positivity of all states in $S_k$.

\subsection{Minimum-error quantum state discrimination}

Suppose there is a box which prepares a quantum system in one of many possible quantum states $\rho_i$, with each preparation happening with a priori probability $p_i$. Then, minimum-error quantum state discrimination corresponds to the measurement procedure through which we obtain the highest probability of success (equivalently, the minimum probability of error) in correctly identifying which state was actually prepared. The figure of merit can be taken to be the optimal probability of guessing correctly,
\[
p_{\mathrm{guess}} = \max_{M_i} \sum_i p_i \tr(M_i\rho_i),
\]
where the maximization is over all generalized measurements (also known as Positive-Operator-Valued Measures (POVMs)) $\{M_i\}$, satisfying $M_i\geq 0$, $\sum_i M_i = \openone$. Notice that, by embedding the probabilities in the definition of unnormalized states $\tilde{\rho_i} = p_i \rho_i$, we can rewrite this as
\[
p_{\mathrm{guess}} = \max_{M_i} \sum_i \tr(M_i\tilde{\rho_i}),
\]

The optimal discrimination strategy, with corresponding optimal probability of success, for two quantum states $\rho_1$ and $\rho_2$, given with probability $p$ and $1-p$, respectively, has been obtained in Ref. \cite{ref:hels}. One has
\bea
\begin{split}
p_{\mathrm{guess}}^{}
&= \frac{1}{2} (1 + \| p \rho_1 - (1-p) \rho_2 \|_1 ) \\
&= \frac{1}{2} (1 + \| \tilde{\rho}_1 - \tilde{\rho}_2 \|_1 ),
\end{split}
\label{eq:hels}
\eea
where $\| \cdot \|_1$ denotes the trace-norm, $\| A \|_1 = \tr \sqrt{A^{\dagger}A }$ for $A\in \B(\H)$. 

\subsection{Minimum-error quantum (sub)channel discrimination }

Suppose that there is a box which applies a quantum operation to a quantum system, and then the resulting system is returned. Let us assume that the box can apply either of two known quantum channels $\Phi_1$ or $\Phi_2$, the first with probability $p$ and the second with probability $1-p$, respectively. This is the scenario of binary channel discrimination, which is particularly relevant when we want to distinguish between an ideal physical evolution and the actual real evolution. For an input state $\rho$, the resulting state is either $\Phi_1 (\rho)$ or $\Phi_{2} (\rho)$; we can then proceed to discriminate the two channels by discriminating such two states. A remarkable property in quantum channel discrimination is that correlations of the input probe with some ancillary system may actually substantially increase the guessing probability~\cite{ref:kitaev}. To be useful, such correlations need to come from probe-ancilla entanglement, rather than just classical correlations. Preparing a state $\rho_{AS} \in S(\H_k\otimes \H_i)$ of a $k$-dimensional ancilla $\H_k$ and of the input system $\H_i$, one has resulting states $\rho_1 = (\i_{k} \otimes \Phi_1) (\rho_{AS})$ and $ \rho _2 = (\i_{k} \otimes \Phi_2)(\rho_{AS})$. Then, the channel applied in the box is found by discriminating between two resulting state $\rho_1$ and $\rho_2$.

While the formula \eqref{eq:hels} already accounts for the optimal final measurement, we still have to optimize over the input state preparation $\rho$, so that the resulting states $\rho_1 = (\i_{k} \otimes \Phi_1) (\rho_{AB})$ and $ \rho _2 = (\i_{k} \otimes \Phi_2)(\rho_{AB})$ are the most distinguishable for the given channels. This is mathematically captured by introducing a norm on Hermitian maps $\Phi \in \T (\H ,\H')$,
\bea
\| \Phi \|^{(k)} = \max_{\rho_{AS} \in S(\H_k \otimes \H)} \|  \i_k \otimes \Phi (\rho_{AS}) \|_1. \label{eq:knorm} 
\eea 
With such a notion for the $k$-norm of linear maps, we can now write the channel distance between $\Phi_1$ and $\Phi_2$ applied with probabilities $p$ and $1-p$, respectively, together with the optimisation over input quantum states that involve a $k$-dimensional ancilla,
\bea
D_{k}^{p}  [ \Phi_1,\Phi_2 ]  & = &  \|  p\Phi_1 - ( 1-p ) \Phi_2  \|^{(k)} \nonumber \\
& = & \max_{\rho_{AS}  } \| \i_k \otimes (p \Phi_1  - (1-p)   \Phi_2 )(\rho_{AS})    \|_1.  ~~ \label{eq:dist}
\eea 
Then, the optimal (both with respect to the choice of input and of final measurement) guessing probability for the two quantum channels, when exploiting a $k$-dimensional ancillary system, is given by 
\bea
p_{\mathrm{guess}}^{\textrm{(k)}}  = \frac{1}{2}(1 + D_{k}^p [\Phi_1 , \Phi_2] ). \label{eq:guessch}
\eea

We remark that the case where no ancilla is used, or where only separable probe-ancilla correlations are present,  is captured by the distance $D_{1}^p = \max_{\sigma\in \mathrm{SEP}} \|  p\i \otimes \Phi_1 (\sigma )  -(1-p)\i \otimes \Phi_2 (\sigma)  \|_1$, where $\mathrm{SEP}$ denotes the set of all separable states.
The other extreme case  is $D_{d}^{p}$, which for $p=1/2$ corresponds to the so-called diamond norm of channels $\| \cdot \|_{\diamond}$, or the norm of complete boundness, $\| \cdot \|_{cb}$, or cb-norm for short~\cite{paulsen2002completely}.

In Ref. \cite{ref:pw}, the usefulness of all entangled states in channel discrimination is shown. Namely, a state $\rho_{AS} \in S(\H^{(A)}\otimes \H^{(S)})$ is entangled if and only if there exists a pair of quantum channels $\Phi_1,\Phi_2 \in \T (\H^{(S)} , \H^{(S')})$ such that 
\bea
&&  \|  \i \otimes \Phi_1 (\rho_{AS})  -\i \otimes \Phi_2 (\rho_{AS})   \|_{1}  \nonumber \\
 && ~~~~~~~ >   \max_{\sigma\in \mathrm{SEP}} \|  \i \otimes \Phi_1 (\sigma )  -\i \otimes \Phi_2 (\sigma)  \|_1.
 \label{eq:piani}
\eea
In terms of the distance measure in Eq. (\ref{eq:dist}), the result can be restated as follows: a state $\rho_{AS} \in S(\H^{(A)}\otimes \H^{(S)})$ is entangled if and only if there exists a pair of channels $\Phi_1$ and $\Phi_2$ such that
\bea
\|\Phi_1 - \Phi_2 \|_{\diamond} & = & D_{d}^{1/2} (\Phi_1 , \Phi_2)  \nonumber \\
& \geq & \|  \i \otimes \Phi_1 (\rho_{AS})  -\i \otimes \Phi_2 (\rho_{AS})   \|_{1} \nonumber \\
&  > &  D_{1}^{1/2} (\Phi_1 , \Phi_2). \label{eq:pw}
\eea
That is, entangled states necessarily improve the discrimination of some pair of channels.

Notice that while Ref. \cite{ref:pw} considered only the case of equal a priory probability $p=1-p=1/2$, it is easy to generalize the result to the case of an arbitrary given $p$, meaning the following: for a fixed an entangled $\rho_{AS}$, for any choice of $0<p<1$, it is possible to find two channels $\Phi_1$ and $\Phi_2$ given with a priori probability $p$ and $1-p$, respectively, such that the probability of correctly guessing the which channel information is increased by the use of $\rho_{AS}$ compared to the case where no-correlations or simply classical correlations are used. In the following, when generalizing the result of Ref. \cite{ref:pw}, we will stick to $p=1/2$, but again an arbitrary probability could be considered.

The task of channel discrimination can be generalized to the task of subchannel discrimination -- and/or its subclass of multichannel discrimination -- where one is tasked with guessing correctly which  branch of the evolution took place among the many possible described by an instrument $\{\tilde{\Lambda}_i\}$. It is important to emphasize that in general subchannel discrimination we cannot imagine that the which-subchannel information is available before the evolution is applied to the probe, but we can in multichannel discrimination.
With the use of a fixed probe-ancilla state $\rho$, the optimal guessing probability is
\[
p_{\mathrm{guess}}(\tilde{\Lambda}_i,\rho) = \max_{M_i} \sum_i p_i \tr(M_i\i\otimes\tilde{\Lambda}_i(\rho)).
\]

\section{Results}

\subsection{Schmidt number and binary channel discrimination}

In this section, we generalize the main result of Ref.~\cite{ref:pw}, by proving that any quantum state $\rho_{AS}$ not in $S_{k}$, that is, with Schmidt number strictly higher than $k$, leads to a an improved channel discrimination probability for some pair of channels acting on the probe system $S$, with respect to what achievable with arbitrary input states in $S_k$.

Namely, we show that $\rho_{AS}$ satisfies $\sn(\rho_{AS})>k$ if and only if there exist a pair of quantum channels $\Phi_1$ and $\Phi_2$ such that
\bea
&&     \| \i \otimes \Phi_1 (\rho) -  \i\otimes \Phi_2 (\rho)  \|_1 \nonumber \\
&  & ~~~~~~~ > \max_{\sigma: \sn(\sigma) \leq k } \| \i \otimes \Phi_1 (\sigma) -  \i \otimes \Phi_2 (\sigma)  \|_1 \label{eq:result}
\eea

The proof of this follows closely the proof in Ref.~\cite{ref:pw}. In particular, we make use of the fact that it is possible to strengthen \eqref{eq:iffSk} to the consideration of trace-preserving maps only:
\begin{multline}
\label{eq:iffSkTP}
\rho_{AS}\in S_k
\iff \\ \i \otimes \Lambda^{(k)}_{\textrm{TP}} (\rho_{AS}) \geq 0 \textrm{ for all TP }k\textrm{-positive maps } \Lambda^{(k)}_{\textrm{TP}},
\end{multline}
or, equivalently,
\begin{multline}
\label{eq:iffSkTPnot}
\rho_{AS}\notin S_k
\iff \\ \i \otimes \Lambda^{(k)}_{\textrm{TP}} (\rho_{AS}) \ngeq 0 \textrm{ for some TP }k\textrm{-positive map } \Lambda^{(k)}_{\textrm{TP}}.
\end{multline}
Notice that, given a normalized state $\rho$ and a trace-preserving map $\Gamma$, $\|\Gamma(\rho)\|_1=1$ if and only if $\Gamma(\rho)\geq 0$, so, if $\Gamma(\rho)\ngeq 0$, then $\|\Gamma(\rho)\|_1>1$. In the following, we consider $\Gamma = \i \otimes \Lambda^{(k)}_{\textrm{TP}}$, that is, the partial action of a trace-preserving $k$-positive map, so that \eqref{eq:iffSkTPnot} can be equivalently stated as
\begin{multline}
\label{eq:iffSkTPnorm}
\rho_{AS}\notin S_k
\iff \\ \|\i \otimes \Lambda^{(k)}_{\textrm{TP}} (\rho_{AS})\| > 1 \textrm{ for some TP }k\textrm{-positive map } \Lambda^{(k)}_{\textrm{TP}}.
\end{multline}

Now, given a trace-preserving $k$-positive map $\m : \B(\H_{d_2}) \rightarrow \B(\H_{d_3})$, one can construct a trace-annihilating map $\ma$ by increasing the output dimension, as
\bea
\ma[A] = \m [A]- (\tr A )~| 0\rangle \langle 0|, \nonumber 
\eea
where $|0\rangle$ is orthogonal to all vectors in $\H_{d_3}$.
As shown in Ref.~\cite{ref:pw}, for a trace-annihilating map $\Lambda_{\mathrm{TA}}$,
 there exist a pair of quantum channels $\Phi_1$ and $\Phi_2$, and also a constant $c_{\Lambda_{\mathrm{TA}}} >0$, such that
\bea
c_{\Lambda_{\mathrm{TA}}} \Lambda_{\mathrm{TA}} = \Phi_1 - \Phi_2. \label{eq:1}
\eea
\\

Thus, we have identified two channels $\Phi_1,\Phi_2$ such that
\[
\Phi_1 - \Phi_2 = c_{\Lambda_{\mathrm{TA}}} ( \m [\cdot]- \tr(\cdot )~| 0\rangle \langle 0|).
\]
Therefore, we have
\begin{align}
&\quad\| \i \otimes (\Phi_1 -\Phi_2) [\rho] \|_1 \nonumber \\ 
& =  c_{\Lambda_{\mathrm{TA}}} \|     \i\otimes \Lambda_{\mathrm{TP}}^{(k)} [\rho] - ( \tr_{\H_{d_2}} \rho ) \otimes | 0\rangle \langle 0|  \|_1 \nonumber \\
& = c_{\Lambda_{\mathrm{TA}}} ( \|      \i\otimes \Lambda_{\mathrm{TP}}^{(k)} [\rho] \|_1 +1). \label{eq:6} 
\end{align}
For any state $\sigma$ having $\sn(\sigma)\leq k$, it holds that
\bea
\| \i \otimes (\Phi_1 -\Phi_2) [\sigma] \|_1 = c_{\Lambda_{\mathrm{TA}}} ( \| \i \otimes \Lambda_{\mathrm{TP}}^{(k)} [\sigma]  \| + 1) = 2 c_{\Lambda_{\mathrm{TA}}}. \nonumber
\eea
while, if $\rho$ is detected by the partial action of $\Lambda_{\mathrm{TP}}^{(k)}$ as having Schmidt number strictly larger than $k$,
\[
\| \i \otimes (\Phi_1 -\Phi_2) [\rho] \|_1 = c_{\Lambda_{\mathrm{TA}}} ( \| \i \otimes \Lambda_{\mathrm{TP}}^{(k)} [\sigma]  \| + 1) > 2 c_{\Lambda_{\mathrm{TA}}}.
\]

\subsection{Schmidt-number robustness and channel discrimination}

While in the preceding section we have considered the case of binary channel discrimination, in this section we focus on multichannel discrimination, finding that every state that has Schmidt number strictly larger than $k$ allows us to identify correctly the which-channel information for some tailored multichannel discrimination problem better than what allowed by any state with Schmidt number $k$ or less. To prove this, we will need some preliminary results.

\begin{prop}
The Schmidt number robustness $R_{S_k}$ is stable under embedding into larger local spaces.
\end{prop}
\begin{proof}
	Let $W_{AB}$ be an optimal witness for the sake of the Schmidt number robustness $R^{d_A,d_B}_{S_k}(\rho_{AB})$ in local dimensions $d_A\geq \rank(\rho_A)$ and $d_B\geq \rank(\rho_B)$, that is, $R^{d_A,d_B}_{S_k}(\rho_{AB})=-\tr(W_{AB}\rho_{AB})$, with $W_{AB}\leq \openone_{d_A}\otimes \openone_{d_B}$ and $\tr(W_{AB}\sigma_{AB}^{(k)})\geq 0$ for all $\sigma_{AB}^{(k)}\in S_k^{d_A,d_B}$, where $S_k^{d_A,d_B}$ is the set of states in $S(\H_{d_A}\otimes\H_{d_B})$ with Schmidt number at most $k$. Consider projector $P_A$ onto $\range(\rho_A)$ and $P_B$ onto $\range(\rho_B)$. It is immediate to check that $W'_{AB} = P_A\otimes P_B W_{AB} P_A \otimes P_B$ satisfies the conditions for a feasible Schmidt number witness in the optimization for  $R^{\rank(\rho_A),\rank(\rho_B)}_{S_k}(\rho_{AB})$; moreover, since $\range(\rho_{AB})\subset\range(\rho_A)\otimes\range(\rho_B)$, one has $\tr(W'_{AB}\rho_{AB})=\tr(W_{AB}\rho_{AB})$.
\end{proof}

\begin{lem}
Through a suitable embedding into a larger local dimension on $A$, an optimal Schmidt number witness $W_{AB}$ can be found for the Schmidt number robustness $R_{S_k}$ that satisfies $W_B=\tr_A W_{AB} \propto \openone_B$.
\end{lem}
\begin{proof}
	Let $\tilde{W}_{AB}$ be optimal for the sake $R_{S_k}$. We will construct construct a new optimal witness $W_{AB}=\tilde{W}_{AB}\oplus\Delta W_{AB}$ such that $\tr_A W_{AB} \propto \openone_B$. Let us define
	\[
	\Delta W_{AB}:=\frac{\openone_{\lceil\|W_B\|_\infty\rceil}^\perp}{\|\tilde{W}_B\|_\infty}\otimes(\|\tilde{W}_B\|_\infty \openone_B-\tilde{W}_B),
	\]
	with $\openone_{\lceil\|W_B\|_\infty\rceil}^\perp$ the identity operator on a $\lceil\|W_B\|_\infty\rceil$-dimensional space orthogonal to the support of $\rho_A$ and $\tilde{W}_{A}$. The operator $\Delta W_{AB}$ is positive semidefinite by construction, so $W_{AB}=\tilde{W}_{AB}\oplus\Delta W_{AB}$ still has non-negative expectation value with respect to states in $S_k$. Also, still by construction,
	\[
	\begin{split}
	\tr_A W_{AB}
	&= \tr_A \tilde{W}_{AB} + \tr_A \Delta W_{AB} \\
	&= \tilde{W}_B +\|\tilde{W}_B\|_\infty \openone_B-\tilde{W}_B \\
	&= \|\tilde{W}_B\|_\infty \openone_B
	\end{split}
	\]
	as required by the claim. Moreover,
	\[
	\tr(\rho_{AB}W_{AB}) = \tr(\rho_{AB}\tilde{W}_{AB})
	\]
	by construction, because of ther orthogonality of the supports of $\Delta W_{AB}$ and $\rho_{AB}$. It remains to be seen that $W_{AB}\leq \openone_{AB}$. This is the case, since for any pure state $\ket{\psi}_{AB}$, we have $\ket{\psi}_{AB} = \sqrt{p}\ket{\psi_1}_{AB}\oplus\sqrt{1-p}\ket{\psi_2}_{AB}$, for $0\leq p\leq 1$, and normalized $\ket{\psi_1},\ket{\psi_2}$ and the same direct sum structure as in $W_{AB}$. Thus,
	\[
	\bra{\psi}W\ket{\psi}=p \bra{\psi_1}\tilde{W}\ket{\psi_1} + (1-p) \bra{\psi_2}\Delta W\ket{\psi_2}\leq 1
	\]
	since $\|\tilde{W}\|_\infty \leq 1$ by assumption and $\|\Delta W\|_\infty \leq 1$ by construction.
\end{proof}
Given the last lemma, we can assume that we are working with local dimensions $d_A,d_B$ into which a given state $\rho_{AB}\in S(\H_{d_A}\otimes\H_{d_B})$ is embedded such that we deal with an optimal witness $W_{AB}$ for $R_{S_k}(\rho_{AB})$ satisfying $-\tr(W_{AB}\rho_{AB})=R_{S_k}(\rho_{AB})$, and $W_B = \tr_B(W_{AB})\propto \openone_B$. We recall also that $W_{AB}\leq \openone_{AB}$, so that we can define $F_{AB} = \openone_{AB}-W_{AB}\geq 0$, with $\tr(F_{AB}\rho_{AB})=1+R_{S_k}(\rho_{AB})$. Since $F_{AB}\geq 0$ and $F_{B} = d_A \openone_B + \tr_A(W_{AB})\propto \openone_B$, we can interpret it as the Choi-Jamiolkowski operator of a completely positive map $c\Lambda^\dagger$ with $c\geq 0$ and $\Lambda^\dagger\in\T(\H_{d_A},\H_{d_B})$ unital, so that $\Lambda\in\T(\H_{d_B},\H_{d_A})$, its dual via $\tr(X\Lambda^\dagger(Y))=\tr(\Lambda(X)Y)$ $\forall X,Y$, is trace-preserving. We thus, for any state $\tau_{AB}$, we have
\[
\tr(F_{AB}\tau_{AB}) = c\tr(\ket{\phi_{d_A}^+}\bra{\phi_{d_A}^+}\i\otimes\Lambda(\tau_{AB}))
\]
where $\ket{\phi_{d_A}^+}$ is the $d_A$-dimensional maximally entangled state, and $\Lambda$ is a completely positive and trace-preserving, hence a quantum channel.

Let us now consider the $d_a^2$ channels $\Gamma_{k,l} = X^kZ^l\Lambda[\cdot](X^kZ^l)^\dagger$, for $k,l=0,\ldots, d_A-1$, with $X$ the unitary shift operator defined by the action $X\ket{n}=\ket{n+1}$ (with addition understood to be modulo $d_a$) on a computational basis $\{\ket{n}\}_{n=0}^{d_A-1}$, and $Z$ the unitary phase operator defined by the action $Z\ket{n}=\exp(i2\pi n/d)\ket{n}$.

We now analyze the success in discriminating these channels in two cases: with the entangled probe-ancilla state $\rho_{AB}$, and with a state $\sigma^{k}\in S_k$. We will consider the case were the \emph{a priori} probability of each of the channels $\Gamma_i$ is the same and equal to $1/d_A^2$.

In the case where we use the probe-ancilla state $\rho_{AB}$, we consider the final POVM to be of the following form:
\[
M_{k,l} = (\openone \otimes X^kZ^l )  \proj{\phi_{d_A}^+} (\openone \otimes X^kZ^l )^\dagger.
\]
It is well known that the POVM elements so defined form an orthonormal basis for $\H_{d_A}\otimes \H_{d_A}$ \cite{werner2001all}.

The probability of guessing correctly is equal exactly to
\begin{align}
&\quad\frac{1}{d_A^2}\sum_{k,l} \tr(M_{k,l}\i\otimes\Gamma_{k,l}(\rho_{AB})) \nonumber\\
& = \frac{1}{d_A^2}\sum_{k,l} \tr(\proj{\phi_{d_A}^+} \i\otimes\Lambda (\rho_{AB})) \nonumber\\
& = \frac{1}{c}\tr(F_{AB}\rho_{AB}) \\
& = \frac{1}{c}(1+R_{S_k}(\rho_{AB})).
\end{align}

We want to upper bound the probability of guessing correctly when using instead $\sigma^{(k)}$ as input, taking into account the use of an arbitrary POVM $\{N_{k,l}\}$. We have 
\begin{align}
&\quad\frac{1}{d_A^2}\sum_{k,l} \tr(N_{k,l}\i\otimes\Gamma_{k,l}(\sigma^{(k)}_{AB})) \nonumber\\
& \leq \frac{1}{d_A^2}\sum_{k,l} \tr(N_{k,l})\|\i\otimes\Gamma_{k,l}(\sigma^{(k)}_{AB})\|_\infty \nonumber\\
& = \frac{1}{d_A^2} \tr(\sum_{k,l}N_{k,l}) \frac{1}{c}\|\i\otimes\Lambda(\sigma^{(k)}_{AB})\|_\infty \nonumber \\
& \leq \frac{1}{c}\frac{1}{d_A^2} \tr(\openone_{d_A^2}) \nonumber \\
& = \frac{1}{c},
\end{align}
where we have used Holder's inequality, the unitary invariance of the operator norm, and the fact that, since $\Lambda$ is a channel, then $\i\otimes\Lambda(\sigma^{(k)}_{AB})$ is a valid state with operator norm less or equal to one.

By considering this particular construction of a channel discrimination problem, we have proven that
\[
\sup \frac{p_{\textrm{guess}}(\rho)}{p^{\textrm{(k)}}_{\textrm{guess}}}\geq1+R_{S_k}(\rho),
\]
where the supremum is over all channel discrimination problems, and the ratio is between the probability of guessing correctly by making use of $\rho_{AB}$ as input, or by making use of an arbitrary state $\sigma^{(k)}$ with Schmidt number less or equal to $k$. On the other hand, by the very definition of robustness
\[
\begin{split}
p_{\textrm{guess}}(\rho)
&= p_{\textrm{guess}}\big((1+R_{S_k}(\rho))\sigma^{{(k)}} - R_{S_k}(\rho)\tau\big)\\
&\leq (1+R_{S_k}(\rho)) p^{(k)}_{\textrm{guess}},
\end{split}
\]
hence we see that actually it is
\[
\sup \frac{p_{\textrm{guess}}(\rho)}{p^{\textrm{(k)}}_{\textrm{guess}}}=1+R_{S_k}(\rho).
\]
A similar about the robustness of entanglement, but with conceptually important differences, was obtained by R. Takagi et al.~\cite{takagi2018}.

\section{Conclusions}

We have have generalized the analysis of Ref. \cite{ref:pw} on the usefulness of entanglement in channel discrimination. We have considered both the case of binary channel discrimination and of multichannel discrimination. In both cases, we have highlighted how, in a sense, the more entangled a bipartite state is, according to the notion of Schmidt number, the larger the set of states it outperforms, and the largest the advantage. We have shown that the latter advantage, in the case of multichannel discrimination, is captured in a very precise and elegant way by the Schmidt number robustness, which generalizes the entanglement robustness. This is remarkable as it gives an exact operational interpretation of the Schmidt number robustness. It is worth noticing that the specific multichannel discrimination task that we conceived is the concatenation of a local, fixed, and deterministic ``filter'', followed by the local action of the same local unitaries used in dense coding. What we are exploiting is the purely quantum effect that a local transformation can generate an entire basis for a bipartite system. Also, if we move away from the framework of channel discrimination, it is fascinating to consider the action of the filtering channel on one system as a simple, single-sided noise model affecting a dense-coding scenario. In such a very simple model, the Schmidt number robustness captures the maximum advantage that a ``highly entangled'' state can give. Finally, we notice that we were able to single out the Schmidt number robustness as advantage factor, while limiting ourselves to considering multiple channels, without the need to consider more general instruments. This suggests that also the quantitative characterization of steering in terms of steering robustness in Ref. \cite{piani2015necessary} might be improved. 

\begin{acknowledgments}
	M. P. thanks F. G. S. L. Brand{\~a}o for useful and inspiring discussions, and for his hospitality during a visit at the University College London, where some of the results presented here were obtained. We thank R. Takagi, B. Regula, G. Adesso and their collaborators for sharing their related preliminary results.
	J. B. is supported by National Research Foundation of Korea (NRF2017R1E1A1A03069961) and the ITRC (Information Technology Research Center) support program (IITP-2018-2018-0-01402) supervised by the IITP (Institute for Information \& communications Technology Promotion). M. P. acknowledges support from European Union's Horizon 2020 Research and Innovation Programme under the Marie Sk{\l}odowska-Curie Action OPERACQC (Grant Agreement No. 661338), and from the Foundational Questions Institute under the Physics of the Observer Programme (Grant No. FQXi-RFP-1601). D. C. was supported by the Polish National Science Centre project 2015/19/B/ST1/03095.
\end{acknowledgments}



%

\end{document}